\documentclass[twocolumn,final]{IEEEtran}
\def\comment#1{}
\usepackage{graphicx,epsfig,color,cite}
\usepackage{ifthen,calc,graphicx,pgfplots}
\usepackage{bm,lscape,rotating,dsfont,enumerate}
\usepackage{amsfonts,amsmath,amssymb}
\newtheorem{theorem}{Theorem}%[section]
\newtheorem{lemma}[theorem]{Lemma}

\def\xB{{\mathbf x}}
\def\yB{{\mathbf y}}
\def\XB{{\mathbf X}}
\def\YB{{\mathbf Y}}
\def\AB{{\mathbf A}}
\def\UB{{\mathbf U}}
\def\VB{{\mathbf V}}
\def\RB{{\mathbf R}}
\def\SigmaB{{\mathbf \Sigma}}
\def\sB{{\mathbf s}}

\def\Cn{{\mathbb C}}

\begin{document}
\setlength{\parskip}{.02in}

\title{Canonical correlation analysis of high-dimensional data with very small sample support}

\author{Yang~Song, Peter~J.~Schreier, David~Ram\'irez, and Tanuj Hasija
        % <-this % stops a space
\thanks{Y.~Song, P.~J.~Schreier, and T. Hasija  are with the Signal and System Theory Group, University of Paderborn, 33098 Paderborn, Germany (e-mail: \{yang.song, peter.schreier, tanuj.hasija\}@sst.upb.de). D.~Ram\'irez is with the Dpto. de Teor\'{i}a de la Se\~{n}al y Comunicaciones, Universidad Carlos III de Madrid, 28911 Legan\'{e}s, Spain (e-mail: ramirezgd@tsc.uc3m.es).}
\thanks{This research was supported by the German Research Foundation (DFG) under grant SCHR 1384/3-1, and the Alfried Krupp von Bohlen und Halbach foundation under its program ``Return of German scientists from abroad.''}
}

\maketitle
\begin{abstract}
This paper is concerned with the analysis of correlation between two high-dimensional data sets 
when there are only few correlated signal components but the number of samples is very small, possibly much smaller than the dimensions of the data. In such a scenario, a principal component analysis (PCA) rank-reduction preprocessing step is commonly performed before applying canonical correlation analysis (CCA). We present simple, yet very effective approaches to the {\em joint} model-order selection of the number of dimensions that should be retained through the PCA step {\em and} the number of correlated signals. These approaches are based on reduced-rank versions of the Bartlett-Lawley hypothesis test and the minimum description length information-theoretic criterion. Simulation results show that the techniques perform well for very small sample sizes even in colored noise.
\end{abstract}

\section{Introduction}
\label{intro}

Correlation analysis based on only small sample support is a challenging task yet with important applications in areas as diverse as biomedicine (e.g. \cite{Lin2006,Correa2010}), climate science (e.g.~\cite{Wallace1992,Shabbar2004}), array processing (e.g. \cite{GeICASSP09}), and others. In this paper, we look at the scenario where the data sets have large dimensions but there are only few correlated signal components. Probably the most common way of analyzing correlation between two data sets is canonical correlation analysis (CCA) \cite{HotellingBio36}. In CCA, the observed data $\xB \in \Cn^n$ and $\yB \in \Cn^m$ are transformed into $p$-dimensional internal (latent) representations ${\bf a} = {\bf S}\xB$ and ${\bf b} = {\bf T}\yB$, where $p = \min(n,m)$, using linear transformations described by the matrices ${\bf S} \in \Cn^{p \times n}$ and ${\bf T} \in \Cn^{p \times m}$. The key idea is to determine ${\bf S}$ and ${\bf T}$ such that most of the correlation between $\xB$ and $\yB$ is captured in a low-dimensional subspace.  

CCA proceeds as follows. First two vectors (``projectors'') ${\bf s}_1 \in \Cn^n$ and ${\bf t}_1 \in \Cn^m$ are determined such that the absolute value of the scalar correlation coefficient $k_1$ between the internal variables $a_1 = {\bf s}_1^T \xB$ and $b_1 = {\bf t}_1^T \yB$ is maximized. The internal variables $(a_1, b_1)$ constitute the first pair of {\em canonical variables}, and $k_1$ is called the first {\em canonical correlation (coefficient)}. The next pair of canonical variables $(a_2, b_2)$ maximizes the absolute value of the scalar correlation coefficient $k_2$ (the second canonical correlation) between $a_2 = {\bf s}_2^T \xB$ and $b_2 = {\bf t}_2^T \yB$, subject to the constraint that they are to be uncorrelated with the first pair. A total of $p$ correlations is determined in this manner, and ${\bf S} = [{\bf s}_1, ..., {\bf s}_p]^T$, ${\bf T} = [{\bf t}_1, ..., {\bf t}_p]^T$. CCA can be performed via the singular value decomposition of the coherence matrix \cite{ScharfSPT0300}
\begin{equation}
\RB_{xx}^{-1/2} \RB_{xy} \RB_{yy}^{-1/2} = {\bf F} {\bf K} {\bf G}^H,\label{equ_CCA}
\end{equation}
where $\RB_{xy}$ is the cross-covariance matrix between $\xB$ and $\yB$, and $\RB_{xx}$ and $\RB_{yy}$ are the auto-covariance matrices of $\xB$ and $\yB$. The canonical correlations $0 \leq k_i \leq 1$ are the singular values, which are the diagonal elements of the diagonal matrix ${\bf K}$. The transformations that generate the latent representations ${\bf a}$ and ${\bf b}$ are then described by ${\bf S} = {\bf F}^H \RB_{xx}^{-1/2}$ and ${\bf T} = {\bf G}^H {\bf R}_{yy}^{-1/2}$. 

In practice, we do not know the covariance matrices and must estimate them from samples. If CCA is performed based on sample covariance matrices, it leads to sample canonical correlations $\hat{k}_i$. If the number of samples $M$ is not significantly larger than the dimensions $m$ and $n$, these $\hat{k}_i$'s can be extremely misleading as they are generally substantially overestimated. Indeed, if $M < m+n$ then $m+n - M$ sample canonical correlations are always identically one, which means that they do not carry any information at all about the true population canonical correlations \cite{PezeshkiACSSC04}. In order to avoid this, we perform a dimension-reduction preprocessing step before applying CCA. The most common type of preprocessing is principal component analysis (PCA). That is, instead of applying (\ref{equ_CCA}) directly to the sample covariance matrices, we first extract a reduced number $r_x$ of components from $\xB$ that account for a large fraction of the total variance in $\xB$. Similarly, we extract $r_y$ components from $\yB$ that account for a large fraction of the total variance in $\yB$. CCA is then performed on the components extracted from $\xB$ and $\yB$. The necessity of a PCA step preceding CCA for small sample sizes was shown in \cite{NadakuditiSSP11} using random matrix theory tools. The paper \cite{NadakuditiSSP11}, however, did not answer the critical question of how to determine $r_x$ and $r_y$ such that the estimated $\hat{k}_i$'s best reflect the true population canonical correlations $k_i$. 

At the same time, a key question in any correlation analysis is how many correlated signals there are. If we had access to the population canonical correlations, we could simply count the number of nonzero $k_i$'s. Since we don't, we need to estimate the number $d$ of correlated signals from the estimated $\hat{k}_i$'s. This is a model-order selection problem. In this paper, we present approaches to {\em jointly} determine, for a PCA-CCA setup, the ranks $r_x$ and $r_y$ of the PCA step and the number $d$ of correlated signals based on extremely small sample support, with $M$ possibly less or even substantially less than $m + n$. These approaches rely on the fact that, while $m$ and $n$ may be very large, the number of correlated signals $d$ is often small. However, a complicating factor of the PCA-CCA setup is that PCA is designed to extract components that account for most of the variance {\em within one} data set, but these components are not necessarily the ones that account for most of the correlation {\em between two} data sets.

In the literature, most of the work on model-order selection deals with either (i) determining the number of signals in a single data set \cite{WaxASSPT0485,NadakuditiSPT0708,Lu2015} or (ii) the number of correlated signals between two data sets, but without a PCA step \cite{BartlettBio41,LawleyBio59,FujikoshiBiometrika79,ZhangSPT0493,ChenIET-RSN1096,GundersonJMA97,StoicaSPT0496}. There is only little work on the {\em joint} model-order selection in a PCA-CCA setup, most of which is rather ad hoc  \cite{Zwick1986a,Hwang2013} and only \cite{Roseveare2015} presents a systematic approach. However, none of these joint PCA-CCA techniques works in the sample-poor case. In the absence of any methodical approach in the sample-poor regime, it is common to use very simple rules of thumb such as ``choose the PCA ranks such that a certain percentage (e.g., 70\%) of the total variance/energy in each data set is retained'' (see, e.g., \cite{Wallace1992}). Needless to say, such rules based on experience only work for specific scenarios.

In general, there are two main approaches to model-order selection: hypothesis tests and information-theoretic criteria. {\em Hypothesis tests} \cite{BartlettBio41,LawleyBio59} are usually series of binary generalized likelihood ratio tests (GLRTs). Starting at $s = 0$, they test whether the model has order $s$ (the null hypothesis) or order greater than $s$ (the alternative). If the null hypothesis is rejected, $s$ is incremented and a new test is run. This proceeds until the null hypothesis is not rejected or the maximum model order is reached. The disadvantage of hypothesis tests is that they require the subjective selection of a probability of false alarm. This can be avoided by using {\em information-theoretic criteria} (ICs) (e.g., \cite{WaxASSPT0485}), which  compute a score as a function of model order. This score is the difference between the likelihood for the observed data, which measures how well the model fits the observed data, and a penalty function. With increasing order there is an increasing number of free parameters, and so the model fit becomes better. In order to avoid overfitting, complex models are penalized by the penalty function, which increases with model order. The best trade-off is achieved when the difference of likelihood and penalty function is maximized. It should be noted that the GLRT and IC methods for model-order selection are actually closely linked \cite{StoicaSPL1110}---a fact that we will exploit, as well.

In this paper, we present approaches to the joint model-order selection in a PCA-CCA setup based on reduced-rank versions of both the Bartlett-Lawley hypothesis test and the minimum description length (MDL) IC \cite{WaxASSPT0485}. As far as we know, these are currently the only techniques capable of handling the combined PCA-CCA approach in the sample-poor regime. An early version of this paper was presented at ICASSP 2015 \cite{SongICASSP15}. 

{%\color{red}
 We would also like to contrast our work with so-called {\em sparse CCA} (e.g., \cite{SparseCCA1,SparseCCA2}). In sparse CCA, a sparsity constraint is placed on the projectors $\sB_i$ and ${\bf t}_i$, which means that each canonical variable $a_i$ or $b_i$ is a linear combination of only a few components in $\xB$ and $\yB$, respectively. While sparse CCA was not proposed to deal with the sample-poor scenario, in principle it can be used as an alternative to PCA-CCA {\em if} there is a priori information that the projectors are sparse. However, in many scenarios of interest (e.g., the applications in biomedicine, climate science, and array processing cited above) there is no justification to assume sparse projectors. When applied to non-sparse problems, sparse CCA will not work well.}

Our program for this paper is as follows. In Section \ref{sec:probform}, we formulate the problem and illustrate the issues that arise when performing CCA based on very small sample sizes and how a combined PCA-CCA approach can address these. We present our approaches based on the hypothesis test in Section \ref{sec:hyptest} and based on the MDL-IC in Section \ref{sec:mdl}. Extensive simulation results are shown in Section \ref{sec:sim}.

\section{Problem Formulation}
\label{sec:probform}

We observe $M$ independent and identically distributed (i.i.d.) sample pairs $\xB_i \in \Cn^n$, $\yB_i \in \Cn^m$ that are drawn from the two-channel measurement model
\begin{align}
{\bf x}&= {\bf A}_x {\bf s}_x + {\bf n}_x, \nonumber \\
{\bf y}&= {\bf A}_y {\bf s}_y + {\bf n}_y. \label{eq:model}
\end{align}
The signals $\sB_x \in \Cn^{d + f_x}$ and $\sB_y \in \Cn^{d + f_y}$ are jointly Gaussian with zero means and cross-covariance matrix
\[ \RB_{s_xs_y} =
\left[\begin{array}{cc}
{\bf diag}(\rho_1\sigma_{x,1}\sigma_{y,1},\ldots,\rho_d\sigma_{x,d}\sigma_{y,d}) & {\bf 0}_{d\times f_y} \\
{\bf 0}_{f_x\times d} & {\bf 0}_{f_x\times f_y}
\end{array}\right], \]
where $\sigma_{x,i}$ is the unknown standard deviation of signal component $s_{x,i}$, $\sigma_{y,i}$ the unknown standard deviation of signal component $s_{y,i}$, and $\rho_i$ the unknown correlation coefficient between $s_{x,i}$ and $s_{y,i}$. Hence, the first $d$ components of $\sB_x$ and $\sB_y$ are correlated, whereas the next ($f_x$, $f_y$) components are independent between $\sB_x$ and $\sB_y$. The correlated components may be stronger or weaker than the independent components. Without loss of generality, we assume the auto-covariance matrices $\RB_{s_xs_x}$ and $\RB_{s_ys_y}$ to be diagonal. The matrices $\AB_x \in \Cn^{n \times (d+f_x)}$ and $\AB_y \in \Cn^{m \times (d+f_y)}$ as well as the dimensions $d$, $f_x$, and $f_y$ are deterministic but unknown. Without loss of generality, $\AB_x$ and $\AB_y$ are assumed to have full column-rank. With all these assumptions, $|\rho_i|$ is the $i$th canonical correlation coefficient $k_i$ between $\sB_x$ and $\sB_y$. The noise vectors ${\bf n}_x \in \Cn^n$ and ${\bf n}_y \in \Cn^m$ are independent of each other, independent of the signals, zero-mean Gaussian, and with {\em unknown} (arbitrary) covariance matrices. 

Compared to the dimensions $m$ and $n$ (which may be very large), we assume that there are only few correlated signals and only few independent signals with variance larger than the correlated signals (but there can be many independent signals with variance smaller than the correlated signals). {%\color{red} 
However, because we {\em do not} assume that the mixing matrices $\AB_x$ and $\AB_y$ in (\ref{eq:model}) are sparse, the cross-covariance matrix $\RB_{xy}$ between the observed vectors $\xB$ and $\yB$ is {\em not sparse} and sparse CCA is generally not suitable for this scenario.}

\begin{figure}
\begin{center}
\includegraphics{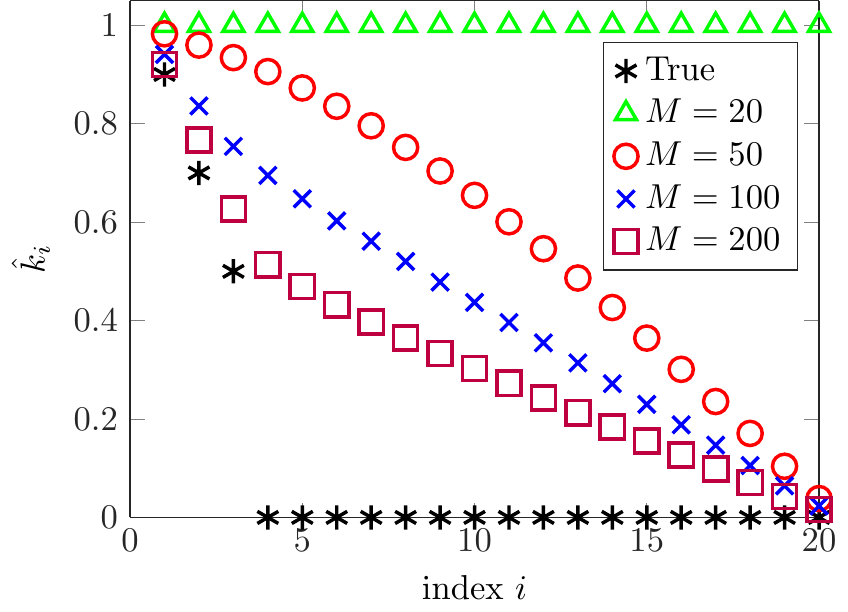}
\caption{Sample canonical correlation coefficients $\hat{k}_i$ for different sample sizes $M$, averaged over $1000$ runs. The are three nonzero population canonical correlations, which are $0.9$, $0.7$, and $0.5$, depicted as $*$. In all cases shown, the true $k_i$'s are significantly overestimated.}
\label{fig:wrongCCs}
\end{center}
\end{figure}

We collect the $M$ sample pairs in data matrices $\XB = [\xB_1, ..., \xB_M]$ and $\YB = [\yB_1, ..., \yB_M]$, from which we compute the sample covariance matrices $\hat{\RB}_{xx} = \XB\XB^H/M$, $\hat{\RB}_{yy} = \YB\YB^H/M$, and $\hat{\RB}_{xy} = \XB\YB^H/M$. In the case of small sample support, the sample canonical correlations $\hat{k}_i$, $i = 1, ..., p$, $p = \min(n,m)$, computed from the sample covariance matrices can be extremely misleading. It has been shown in \cite{PezeshkiACSSC04} that when $M < m+n$, at least $m+n - M$ sample canonical correlations will be identically one regardless of the two-channel model that generates the data samples. In such a small sample scenario, the $\hat{k}_i$'s cannot be used to infer the number of correlated signals. But even in the case with $M$ greater (but not substantially greater) than $m+n$ the sample $\hat{k}_i$'s are generally significantly overestimated. This is shown in Fig. \ref{fig:wrongCCs}, which displays the sample canonical correlations for a model of dimension $m = n = 20$, with $d = 3$ correlated components and $f_x = f_y = 0$ independent components for different sample sizes $M$. Even for $M = 200$, where the number of samples is ten times the dimension of the system, the $\hat{k}_i$'s for $i \geq 4$ are quite wrong, and it is impossible from visual inspection to determine the number of correlated components.

This motivates the use of a rank-reduction preprocessing step. The most common type of preprocessing is PCA, and a combined PCA-CCA approach is the setup that we consider in our paper. So let us investigate what effect rank reduction has on the estimated canonical correlations. The PCA step retains those $r_x$ and $r_y$ components in $\XB$ and $\YB$, respectively, that account for most of their total variance. These components can be computed as follows. We first determine the singular value decompositions (SVDs) of the data matrices $\XB = \UB_x \SigmaB_x \VB_x^H$ and $\YB = \UB_y \SigmaB_y \VB_y^H$. Then the reduced-rank PCA descriptions of $\XB$ and $\YB$ are
\begin{align}
\XB_{r_x} & = \UB_x^H(:,1:r_x) \XB \in \Cn^{r_x \times M}, \nonumber \\
\YB_{r_y} & = \UB_y^H(:,1:r_y) \YB \in \Cn^{r_y \times M}, \label{equ:XrYr}
\end{align}
where ${\bf U}_x(:,1:r_x)$ denotes the $n\times r_x$ matrix containing the first $r_x$ columns of $\UB_x$, which are associated with the largest $r_x$ singular values of ${\bf X}$, and ${\bf U}_y(:,1:r_y)$ denotes the $m\times r_y$ matrix containing the first $r_y$ columns of $\UB_y$, which are associated with the largest $r_y$ singular values of ${\bf Y}$. Now let $\widetilde{\RB}_{xx} = \XB_{r_x}\XB^H_{r_x}/M$, $\widetilde{\RB}_{yy} = \YB_{r_y}\YB_{r_y}^H/M$, and $\widetilde{\RB}_{xy} = \XB_{r_x}\YB_{r_y}^H/M$ be the sample covariance matrices from the reduced-dimensional PCA descriptions. The corresponding estimated canonical correlations $\hat{k}_i(r_x,r_y)$ may be computed as the singular values of the reduced-dimensional sample coherence matrix, which is \cite{PezeshkiACSSC04}
\begin{align}
& {\bf{\widetilde R}}_{xx}^{-1/2} {\bf{\widetilde R}}_{xy} {\bf{\widetilde R}}_{yy}^{-1/2} \nonumber \\
=&\UB^H_x(:,1:r_x) \UB_x \left(\SigmaB_x\SigmaB_x^H\right)^{-1/2} \UB_x^H \UB_x(:,1:r_x) \nonumber \\
&\times \UB^H_x(:,1:r_x) \UB_x \SigmaB_x \VB_x^H \VB_y \SigmaB_y^H \UB_y^H \UB_y(:,1:r_y) \nonumber \\
&\times \UB^H_y(:,1:r_y) \UB_y \left(\SigmaB_y\SigmaB_y^H\right)^{-1/2} \UB_y^H \UB_y(:,1:r_y) \nonumber \\
=&\left[{\bf I}_{r_x},{\bf 0}_{r_x\times(m-r_x)}\right] \left(\SigmaB_x\SigmaB_x^H\right)^{-1/2} \left[{\bf I}_{r_x},{\bf 0}_{r_x\times(m-r_x)}\right]^H \nonumber \\
&\times \left[{\bf I}_{r_x},{\bf 0}_{r_x\times(m-r_x)}\right] \SigmaB_x \VB_x^H \VB_y \SigmaB_y^H \left[{\bf I}_{r_y},{\bf 0}_{r_y\times(m-r_y)}\right]^H \nonumber \\
&\times \left[{\bf I}_{r_y},{\bf 0}_{r_y\times(m-r_y)}\right] \left(\SigmaB_y\SigmaB_y^H\right)^{-1/2} \left[{\bf I}_{r_y},{\bf 0}_{r_y\times(m-r_y)}\right]^H \nonumber \\
=&
\SigmaB_x^{-1}(1:r_x,1:r_x)\left[\SigmaB_x(1:r_x,1:r_x),{\bf 0}_{r_x\times(M-r_x)}\right]\VB_x^H \VB_y
\nonumber \\
& \times\left[\SigmaB_y(1:r_y,1:r_y),{\bf 0}_{r_y\times(M-r_y)}\right]^H\SigmaB_y^{-1}(1:r_y,1:r_y) \nonumber\\
=&
\left[{\bf I}_{r_x},{\bf 0}_{r_x\times(M-r_x)}\right]
\VB_x^H \VB_y
\left[{\bf I}_{r_y},{\bf 0}_{r_y\times(M-r_y)}\right]^H \nonumber\\
=&
\VB_x^H(:,1:r_x) \VB_y(:,1:r_y).\label{equ:VxVy}
\end{align}
The thus computed canonical correlations $\hat{k}_i(r_x,r_y)$, $i = 1, ..., r$, $r = \min(r_x,r_y)$, depend on the ranks $r_x$ and $r_y$. As seen in (\ref{equ:VxVy}), the $i$th estimated canonical correlation $\hat{k}_i(r_x,r_y)$ can be found as the $i$th largest singular value of ${\bf V}^H_x(:,1:r_x) {\bf V}_y(:,1:r_y)$, where $\VB_x$ and $\VB_y$ are the matrices of right singular vectors of $\XB$ and $\YB$, respectively. To avoid defective unit  canonical correlations, we must choose $r_x+r_y \leq M$ and $\max(r_x,r_y) \leq p$. This, however, does not tell us what the optimum choices for $r_x$ and $r_y$ are such that the $\hat{k}_i(r_x,r_y)$'s are as close to the true canonical correlations as possible. 

Intuitively, it seems that $r_x$ and $r_y$ should be chosen large enough to capture as much of the correlated signal components as possible without including too much noise. If the correlated components are weaker than some of the independent components, this will inevitably mean that the PCA preprocessing step must also keep those stronger independent components. On the other hand, if the correlated components are also the strongest components, it would be better if the PCA step got rid of the independent components. Hence, without noise, $r_x$ would ideally be chosen between $d$ and $d+f_x$, and $r_y$ between $d$ and $d+f_y$. With noise, the ranks $r_x$ and $r_y$ may also fall outside of these ranges, depending on the properties of the noise and the relative strengths of the signals.

It can be shown using Cauchy's interlacing theorem (the result is presented as Lemma \ref{lemma_interlace} in Appendix \ref{app:interlace}) that increasing the ranks of the PCA steps increases {\em every} estimated canonical correlation coefficient. Hence, choosing too large an $r_x$ or $r_y$ will lead to estimated canonical correlations that are greater, possibly significantly greater, than the true canonical correlations. On the other hand, if $r_x$ and $r_y$ are not large enough, then the rank-reduced representation does not contain all of the correlated components, and thus the estimated canonical correlations can be too small. 

\begin{figure}
\begin{center}
\includegraphics{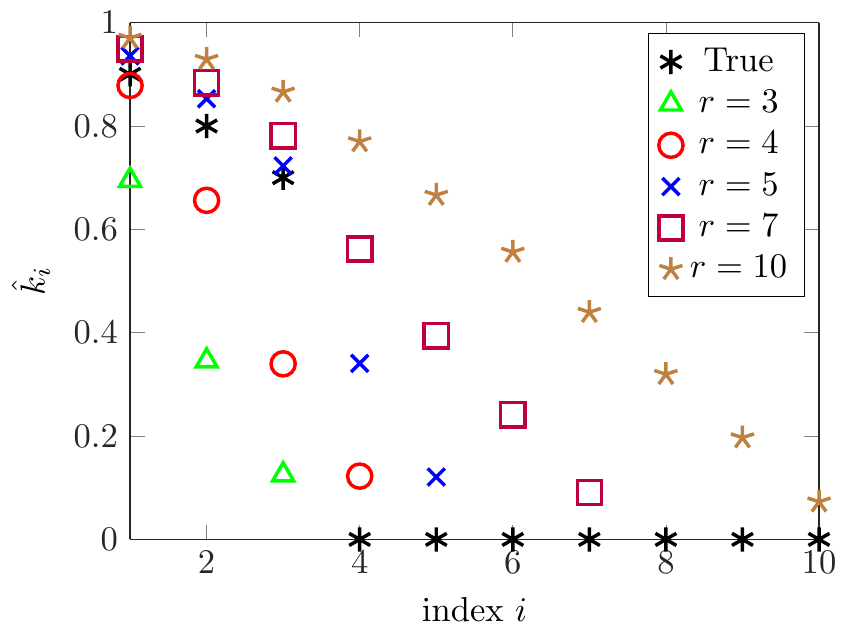}
\caption{Effect of rank reduction on the estimated canonical correlations $\hat{k}_i(r)$, averaged over $1000$ runs. The are $d = 3$ correlated signal components with population canonical correlation coefficients $0.9$, $0.8$, and $0.7$ (depicted as $*$), and $f_x = f_y = 2$ stronger independent signal components.  For $r > 5$, the canonical correlation coefficients are all overestimated. For $r < 5$, the nonzero coefficients are underestimated. The ranks of the PCA steps for $\xB$ and $\yB$ are the same: $r = r_x = r_y$.}
\label{fig_cck}
\end{center}
\end{figure}

These considerations can be illustrated by the following example, where $M = 30$ and $m = n = 20$. There are $d = 3$ correlated signals (each with variance 1.5) and $f = f_x = f_y = 2$ independent signals (each with variance 5). Since the independent signals are stronger than the correlated signals (and the numbers $f_x$ and $f_y$ of independent signals in $\sB_x$ and $\sB_y$ are identical), we would expect $r_x = r_y = d + f = 5$ to be the optimum rank for the PCA step. Indeed, Fig. \ref{fig_cck} shows that choosing $r = r_x = r_y$ greater than 5 leads to $\hat{k}_i$'s that are too large, whereas $r$ less than 5 leads to $\hat{k}_i$'s that are too small. While the exact relationships depend on the variances of signal and noise components and the correlation coefficients, the principle observed here generalizes to other settings.

%The central goal of our paper is thus to find rules for {\em jointly} selecting the best ranks ($r_x$, $r_y$) for the PCA steps and determining the number $d$ of correlated components in the scenario where the sample support is very small, possibly $M < m+n$.

\section{Order selection based on hypothesis test}
\label{sec:hyptest}

\subsection{Traditional test}

In the case of sufficient number of samples, the traditional hypothesis test \cite{BartlettBio41,LawleyBio59} for determining the number $d$ of correlated components between $\xB$ and $\yB$ is a series of binary hypothesis tests. Starting with $s = 0$, it tests the null hypothesis $H_0$: $d = s$ versus the alternative hypothesis $H_1$: $d > s$. If $H_0$ is rejected, $s$ is incremented and a new test is run. This proceeds until $H_0$ is not rejected or $s = p = \min(n,m)$ is reached. 

The binary test in \cite{BartlettBio41,LawleyBio59} is a generalized likelihood ratio test (GLRT) of $H_0$ vs. $H_1$. For a given number $s$ of correlated signals, let $\Omega_s$ denote the parameter space of the model, which consists of the  auto- and cross-covariance matrices. The maximum value of the log-likelihood function for a given number $s$ of correlated signals, maximized over the parameter space $\Omega_s$, is \cite{StoicaSPT0496}
\begin{equation}
\ell_{\max}(\XB,\YB|\Omega_s) = -M \ln \prod_{i=1}^s \left(1 - \hat{k}_i^2 \right).
\end{equation}
Canonical correlation coefficients close to 1 are strong evidence of correlation between $\xB$ and $\yB$ and thus lead to large $\ell_{\max}$. Now let $\Omega_{d > s}$ denote the parameter space of all models where the assumed number of correlated signals $d$ is greater than $s$. The generalized log-likelihood ratio for testing $H_0$ vs. $H_1$ is \cite{ChenIET-RSN1096}
\begin{align}
\Lambda(n,m,s) & = \ell_{\rm max}({\bf X},{\bf Y}|\Omega_{d=s})-\ell_{\rm max}({\bf X},{\bf Y}|\Omega_{d>s}) \nonumber \\
& = \ell_{\rm max}({\bf X},{\bf Y}|\Omega_{s})-\ell_{\rm max}({\bf X},{\bf Y}|\Omega_{p}) \nonumber \\
& = M \ln \prod_{i=s+1}^p \left(1 - \hat{k}_i^2\right), \label{equ:GLRT}
\end{align}
where the second identity follows from the fact that the maximum of the likelihood function, under the constraint $d > s$, occurs when the model has the most degrees of freedom, i.e., for $d = p = \min(n,m)$. The cross-covariance matrix has $N_\Omega(n,m,s) = 2s(m+n-s)$ degrees of freedom \cite{StoicaSPT0496}. Wilks' theorem \cite{Wilks1938} says that $-2 \Lambda(n,m,s)$ is asymptotically (as $M \rightarrow \infty$) $\chi^2$-distributed with degrees of freedom equal to the difference of the dimensions of the parameter spaces $\Omega_p$ and $\Omega_s$:\footnote{In this difference, only the degrees of freedom associated with the cross-covariance matrix matter.} 
\begin{align}
N_\Lambda(n,m,s) & = N_\Omega(n,m,p) - N_\Omega(n,m,s) \nonumber \\
& = 2p(m + n - p) - 2s(m + n - s) \nonumber \\
& = 2(m - s)(n - s)
\end{align} 
For finite $M$, the closeness of the $\chi^2$-approximation may be improved by replacing $-2\Lambda(n,m,s)$ with the Bartlett-Lawley statistic \cite{BartlettBio41,LawleyBio59}
\begin{align}
C(n,m,s) = & - \! 2 \left( M-s-\frac{m+n+1}{2}+\sum_{i=1}^{s} \hat{k}_i^{-2} \right) \nonumber \\
& \times \ln\prod_{i=s+1}^{p}\left(1-\hat{k}_i^2\right).\label{cs1}
\end{align}
This correction makes the moments of the test statistic equal to the moments of the $\chi^2$-distribution. As long as $M$ is large compared to $m$ and $n$, the statistic $C(n,m,s)$ is generally very close to a $\chi^2$-distribution. {%\color{red}
 Note that this is independent of the covariance matrix of the noise, since it is not used anywhere in the derivation.} This allows computation of a test threshold $T(n,m,s)$ for a given probability of false alarm.

\subsection{Test with PCA preprocessing}

Instead of running the test directly on $\XB$ and $\YB$, we would like to apply the test to the reduced-rank PCA descriptions $\XB_{r_x}$ and $\YB_{r_y}$ obtained in (\ref{equ:XrYr}). By performing PCA on $\xB$ and $\yB$, we create a new reduced-rank two-channel model:
\begin{align}
{\bf x}_{r_x}& ={\bf U}^H_x(:,1:r_x){\bf x} \nonumber \\
&={\bf U}^H_x(:,1:r_x){\bf A}_x {\bf s}_x + {\bf U}^H_x(:,1:r_x){\bf n}_x \nonumber \\
&=\widetilde{{\bf A}}_x {\bf s}_x + \widetilde{{\bf n}}_x, \nonumber\\
{\bf y}_{r_y}& ={\bf U}^H_y(:,1:r_y){\bf y} \nonumber \\ &={\bf U}^H_y(:,1:r_y){\bf A}_y {\bf s}_y + {\bf U}^H_y(:,1:r_y){\bf n}_y \nonumber \\
&=\widetilde{{\bf A}}_y {\bf s}_y + \widetilde{{\bf n}}_y. \label{eq:model2}
\end{align}
In this model, the new matrices $\widetilde{\AB}_x$ and $\widetilde{\AB}_y$ have full rank because $\AB_x$ and $\AB_y$ are assumed to have full rank. With the PCA preprocessing the GLRT statistic is
\begin{equation}
\Lambda(r_x,r_y,s) = M \ln \prod_{i=s+1}^r \left(1 - \hat{k}_i^2(r_x,r_y) \right), \label{equ:Lambda}
\end{equation}
and the Bartlett-Lawley statistic is
\begin{align}
C(r_x,r_y,s) = & - \! 2\left( M-s-\frac{r_x+r_y+1}{2}+\sum_{i=1}^{s} \hat{k}_i^{-2}(r_x,r_y) \right)  \nonumber \\
& \times  \ln\prod_{i=s+1}^{r}\left(1-\hat{k}^2_i(r_x,r_y)\right)  \label{cs2}
\end{align}
for $s=0,\ldots,r-1$ with $r = \min(r_x,r_y)$. The challenge in the reduced-rank version of the hypothesis test is thus to {\em jointly} determine the best ranks $r_x,r_y$ of the PCA steps and the number $d$ of correlated signals. As long as the number of samples $M$ is large compared to the minimum PCA dimension $r = \min(r_x,r_y)$ but $r_x$ and $r_y$ are not too small (which we will explain in the next paragraph), the new test statistic $C(r_x,r_y,d)$ under $H_0: d = s$ is still approximately $\chi^2$-distributed with $2(r_x - d)(r_y - d)$ degrees of freedom. We denote by $r_{\max}$ the largest $r$ for which the $\chi^2$-distribution holds well enough. Of course, requiring $M$ to be large with respect to $r$ is a much more relaxed condition than requiring $M$ to be large with respect to the dimensions $n$ and $m$. This is because $r_x$ and $r_y$ do not have to be chosen greater (unless there are strong noise components) than $d + f_x$ and $d + f_y$, respectively, which are usually much smaller than $n$ and $m$. 

There is, however, a complication. By applying PCA to $\xB$ and $\yB$, we might eliminate some of the correlated components if the PCA ranks $r_x$ and $r_y$ are not chosen large enough. If this is the case, then the number of correlated components $\tilde{d}$ in the {\em reduced-rank descriptions} $\xB_{r_x}$ and $\yB_{r_y}$ will be {\em smaller} than the number of correlated components $d$ between $\xB$ and $\yB$. As a consequence, $C(r_x,r_y,d)$ will no longer resemble a $\chi^2$-distribution. Instead, $C(r_x,r_y,\tilde{d})$ with $\tilde{d} < d$ will now be approximately $\chi^2$. By choosing $r_x$ and $r_y$ not large enough it thus becomes likely that the null hypothesis ``there are $\tilde{d}$ correlated signals'' is not rejected, thus deciding for a smaller number $\tilde{d}$ than the true $d$.

We are now getting closer to writing down a rule for jointly selecting $r_x$, $r_y$, and $d$. In order to motivate this rule, we summarize the preceding discussion: Provided the PCA ranks $r_x$ and $r_y$ are chosen sufficiently large to capture all correlated components while $r$ is still small compared to $M$, i.e., $r \leq r_{\max}$, the statistic $C(r_x,r_y,d)$ in (\ref{cs2}) is approximately $\chi^2$ {%\color{red} 
(again irrespective of the noise covariance matrix)}. This means that in a series of binary tests of $H_0 : d = s$ vs. $H_1 : d > s$ (testing all values of $s$ starting from $0$ until $H_0$ is not rejected or the maximum $s = r_{\max}$ is reached) $d$ would generally not be {\em overestimated}. It is likely, however, to be {\em underestimated}, if $r_x$ and $r_y$ are not chosen large enough. If $r_x$ and $r_y$ are too small, then the reduced-rank PCA descriptions do not capture all of the correlated components and thus the series of binary tests would decide for too small a $d$. This reasoning motivates the following decision rule.

\noindent {\bf Detector 1 (``max-min detector''):} {\em Choose
\begin{equation}
{\hat d} = \!\!\!
\underset{
\{r_x,r_y\}=1,\ldots,r_{\max}
}{\max}~
\underset{s=0,\ldots,r-1}{{\min}}\left\{s:C(r_x,r_y,s) < T(r_x,r_y,s)\right\} \label{equ:maxmin-det}
\end{equation}
and choose the $r_x$ and $r_y$ that lead to $\hat{d}$ as the PCA ranks.} In (\ref{equ:maxmin-det}) the ${\min}$-operator chooses the smallest $s$ such that the statistic $C(r_x,r_y,s)$ falls below the threshold $T(r_x,r_y,s)$, which ensures a given probability of false alarm. If there is no such $s$, then it chooses $s = r$. This step is similar to the traditional test, except that $T(r_x,r_y,s)$ depends on $r_x$ and $r_y$. The rule (\ref{equ:maxmin-det}) is based on the fact that if $r_x$ and $r_y$ are not chosen optimally, the min-step might return a number smaller than $d$. Hence, the min-step is performed for all $r_x$ and $r_y$ from 1 up to $r_{\max}$, and the maximum result is chosen as $\hat{d}$.

\subsection{Example}

\begin{figure}
\begin{center}
\includegraphics{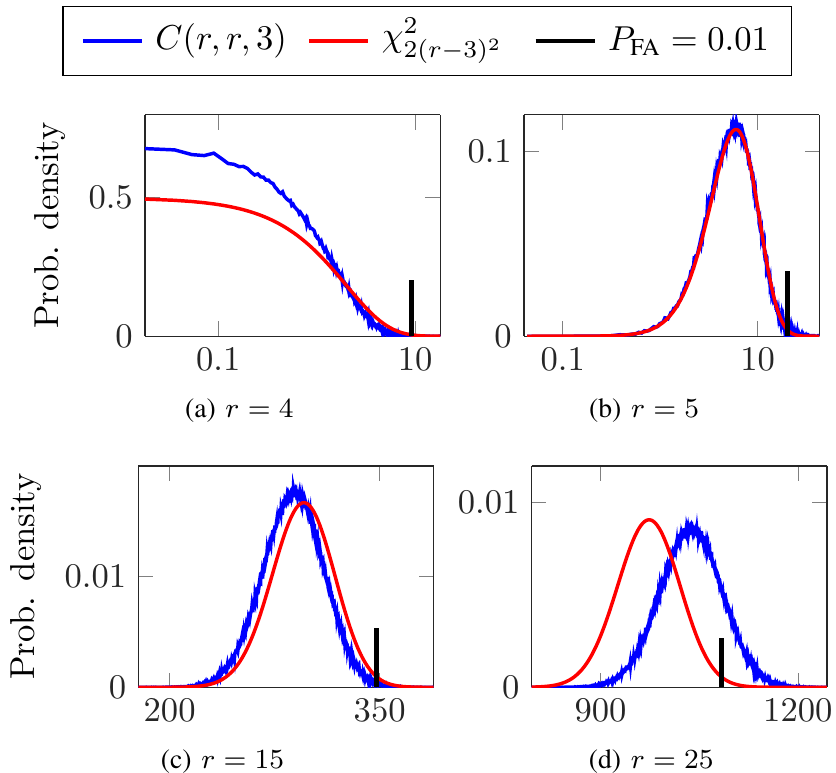}
\caption{Histogram of the test statistic $C(r,r,3)$ (in blue) and the probability density function of a $\chi^2$-distribution with $2(r-3)^2$ degrees of freedom (in red), for $s = d = 3$ and different PCA ranks $r = r_x = r_y$. Histograms are computed from $10^6$ independent trials. Also shown as vertical lines are the thresholds $T(r,r,3)$ for a probability of false alarm $P_{\rm FA} = 0.01$. The horizontal axis uses a logarithmic scale.}
\label{fig:chi2}
\end{center}
\end{figure}

We will use an example to illustrate both the closeness of the $\chi^2$-approximation and the idea of the max-min detector. We consider a scenario with $m = n = 100$, $d = 3$ correlated signals, $f = f_x = f_y = 2$ stronger interfering signals, and $M = 50$ samples. The noise variance is chosen small compared to the signal variances. For $s = d = 3$ and $r = r_x = r_y$, Fig.~\ref{fig:chi2} compares histograms of the statistic $C(r,r,3)$ with the probability density function of a $\chi^2$-distribution with $2(r_x - d)(r_y - d) = 2(r - 3)^2$ degrees of freedom. As long as $r$ is large enough to capture all correlated components (which is the case for $r \geq d + f = 5$ since the independent signals are stronger than the correlated signals) but small compared to $M$, the statistic $C(r,r,3)$ is very well approximated by the $\chi^2$-distribution. This can be seen in subplots (b) $r = 5$ and (c) $r = 15$ (where we start to notice some divergence between the statistic and its approximation). Subplot (d) shows $r = 25$, which is not small enough with respect to $M = 50$. Here the $\chi^2$-distribution is no longer a good approximation of the test statistic.

On the other hand, if $r < 5$ then the PCA step eliminates some correlated components. This can be observed in subplot (a) for $r = 4$, where the histogram of $C(4,4,3)$ does not approximate a $\chi^2$-distribution. Because the PCA steps with $r_x = r_y = 4$ keep the two stronger independent signals and only two of the three weaker correlated signals,
the reduced-rank PCA descriptions $\xB_{r_x}$ and $\yB_{r_y}$ only have $\tilde{d} = 2$ correlated signals rather than $d = 3$. It can be observed in Fig. \ref{fig:notchi2} (b) that $C(4,4,2)$ indeed well approximates a $\chi^2$-distribution with $2(r - \tilde{d})^2 = 2(4 - 2)^2 = 8$ degrees of freedom.

\begin{figure}
\begin{center}
\includegraphics{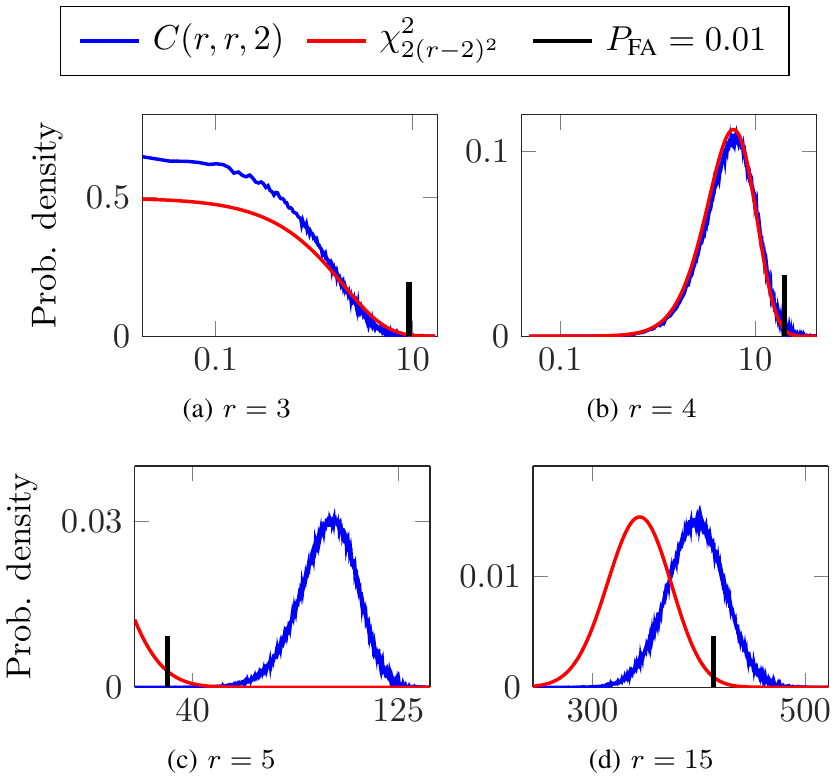}
\caption{Histogram of the test statistic $C(r,r,2)$ (in blue) and the probability density function of a $\chi^2$-distribution with $2(r-2)^2$ degrees of freedom (in red), for $d = 3$ but $s = 2$ and different PCA ranks $r = r_x = r_y$. The vertical lines are the thresholds $T(r,r,2)$ for a probability of false alarm $P_{\rm FA} = 0.01$. The horizontal axis uses a logarithmic scale.}
\label{fig:notchi2}
\end{center}
\end{figure}

So let us look at how the max-min detector would proceed in this example. To illustrate this, we again consider Figs.~\ref{fig:chi2} and \ref{fig:notchi2}, which compare histograms of $C(r,r,s)$ with $\chi^2$-distributions with $2(r - s)^2$ degrees of freedom for $s = d = 3$ (Fig. \ref{fig:chi2}) and $s = 2$ (Fig. \ref{fig:notchi2}). Also shown in these figures are the thresholds $T(r,r,s)$ for a probability of false alarm $P_{\rm FA} = 0.01$. According to (\ref{equ:maxmin-det}), for given $r_x$ and $r_y$, the detector needs to find the minimum $s$ (between $0$ and $r$) such that the statistic $C$ falls below the threshold $T$. Consider first $r_x = r_y = 4$, which is too small because the PCA steps eliminate one correlated component. From Fig. \ref{fig:notchi2} (b), we see that it is likely that $C(4,4,2)$ falls below $T(4,4,2)$, which means that for $r_x = r_y = 4$, the min-step of the detector would likely return too small a number of correlated signals ($s = 2$).\footnote{If we plotted the test statistics and thresholds also for $s = 0$ and $s = 1$ we would see that it is unlikely that a value $s < 2$ would be chosen.} 

Now consider $r_x = r_y = 5$, which is large enough so that the PCA steps capture all correlated components. It can now be observed in Fig. \ref{fig:notchi2} (c) that for $s = 2$, $C(5,5,2)$ will likely not fall below $T(5,5,2)$.\footnote{As before, if we plotted the test statistics also for $s = 0$ and $s = 1$, we would see that it is even less likely that $C(5,5,s)$ falls below $T(5,5,s)$ if $s < 2$.} On the other hand, Fig. \ref{fig:chi2} (b) shows that it is likely that $C(5,5,3)$ falls below $T(5,5,3)$, hence returning $s = 3$ in the min-step of the detector.

Finally, consider $r_x = r_y = 15$, which is larger than needed to capture all correlated components. If $r_x$ and $r_y$ are too large then it becomes increasingly difficult, as can be observed in Fig. \ref{fig_cck}, to distinguish between the sample correlation coefficients that are associated with the correlated signals and those that are not. The min-step of the detector would still not generally overestimate $d$ (because the $\chi^2$-approximation remains valid under $H_0$) but it might {\em underestimate} it. This becomes clear from looking at Fig. \ref{fig:notchi2} (d), which shows that there is a rather high chance that the min-step would select $s = 2$. However, an underestimating min-step is not a problem for the max-min detector because it selects the maximum of all min-step results.

\section{Order selection based on information-theoretic criterion}
\label{sec:mdl}

A disadvantage of the hypothesis testing approach to order selection is the requirement of selecting a probability of false alarm $P_{\rm FA}$. Setting $P_{\rm FA}$ too high will lead to a detector that tends to overfit, setting it too low will generally underfit. Achieving the best performance thus requires the right trade-off. In this section, we present two alternative approaches that do not require the manual selection of a threshold and are based on the minimum description length (MDL)-IC. The first approach will remain a hypothesis test but with automatic $P_{\rm FA}$-selection exploiting a link between the GLRT and the IC for model-order selection. The second approach will be a max-min detector based directly on the MDL-IC.

\subsection{Setting the threshold based on the MDL-IC}

The MDL-IC for selecting the number of correlated signals in two data sets (without PCA steps) is \cite{StoicaSPT0496}
\begin{align}
I_{\rm MDL}(n,m,s)&=
-\ell_{\rm max}(\XB,\YB|\Omega_{s})
+\frac{1}{2}\ln(M)N_\Omega(n,m,s)\nonumber\\
&= M\log\prod_{i=1}^{s}\left(1-\hat{k}_i^{2}\right)+\ln(M)s(m+n-s).\label{eq:itc1}
\end{align}
In this expression, the second term is the penalty term that depends on the degrees of freedom of the model\footnote{As before, only the degrees of freedom associated with the cross-covariance matrix are considered because the degrees of freedom associated with the auto-covariance matrices do not depend on $s$. Hence, they do not matter in the following optimization problems.} and thus penalizes overly complex models. The model order chosen is the value of $s$ for which $I_{\rm MDL}(n,m,s)$ is minimized. The reduced-rank version of (\ref{eq:itc1}), which accounts for the PCA steps, is
\begin{align}
I_{\rm MDL}&(r_x,r_y,s) \nonumber \\ 
% & = -\ell_{\rm max}(\XB_{r_x},\YB_{r_y}|\Omega_{s}) + \frac{1}{2}\ln(M)N_\Omega(r_x,r_y,s) \nonumber\\
& = M\ln\prod_{i=1}^{s}\left(1-\hat{k}_i^{2}(r_x,r_y)\right) +\ln(M)s(r_x+r_y-s).\label{eq:itc2}
\end{align}
As has been noted in \cite{StoicaSPL1110}, there is the following connection between the MDL-IC and the log-likelihood ratio of the reduced-rank GLRT $H_0 : d = s$ vs. $H_1 : d > s$:
\begin{align}
I_{\rm MDL}&(r_x,r_y,r) - I_{\rm MDL}(r_x,r_y,s) \nonumber \\
& = \Lambda(r_x,r_y,s) + \ln(M)N_\Lambda(r_x,r_y,s) \label{equ:equivalence}
\end{align}
with $N_\Lambda(r_x,r_y,s) = (r_x - s)(r_y - s)$. When choosing between model orders $s$ and $r$ based on the MDL-IC, we decide for model order $s$ if $I_{\rm MDL}(r_x,r_y,r) > I_{\rm MDL}(r_x,r_y,s)$. Because of (\ref{equ:equivalence}) we can implement this decision rule also based on the GLRT. We decide for model order $s$ rather than a model order greater than $s$ if
\begin{equation}
\Lambda(r_x,r_y,s) > \underbrace{-\ln(M)(r_x - s)(r_y - s)}_{\displaystyle T_{\rm MDL}(r_x,r_y,s)}. \label{equ:TMDL}
\end{equation}
The term on the right-hand side of this inequality is thus the threshold for the GLRT, which is determined based on the MDL-IC. Note that it is unnecessary to apply the Bartlett-Lawley correction because this would amount to multiplying both sides of the inequality (\ref{equ:TMDL}) with the same factor. Thus, we obtain the following max-min decision rule in terms of $\Lambda(r_x,r_y,s)$ rather than $C(r_x,r_y,s)$.\\
\noindent {\bf Detector 2 (max-min detector with threshold set by MDL-IC):} {\em Choose
\begin{align}
{\hat d} = 
\underset{
\{r_x,r_y\}=1,\ldots,r_{\max}
}{\max}~
\underset{s=0,\ldots,r-1}{{\min}} \{s & :\Lambda(r_x,r_y,s) \nonumber \\
& > T_{\rm MDL}(r_x,r_y,s)\},
\end{align}
where $\Lambda(r_x,r_y,s)$ is given in (\ref{equ:Lambda}) and $T_{\rm MDL}(r_x,r_y,s)$ is given in (\ref{equ:TMDL}), and choose the $r_x$ and $r_y$ that lead to $\hat{d}$ as the PCA ranks.}

\subsection{Min-MDL detector}
\label{sec_minMDL}

Another approach that does not require the selection of $P_{\rm FA}$ applies the max-min idea directly to the MDL-IC. Let us first write down the decision rule and interpret it afterwards. \\
\noindent {\bf Detector 3 (``max-min MDL-IC detector''):} {\em Choose 
\begin{equation}
\hat{d} =  \underset{
\{r_x,r_y\}=1,\ldots,r_{\max}
}{\max}~
\underset{s=0,\ldots,r-1}{{\rm argmin}} I_{\rm MDL}(r_x,r_y,s) \label{equ:maxmin-MDL}
\end{equation}
and choose the $r_x$ and $r_y$ that lead to $\hat{d}$ as the PCA ranks.}

In order to understand this detector, we note, based on the discussion in the preceding subsection, that
\begin{align}
\underset{s=0,\ldots,r-1}{{\rm argmin}} & I_{\rm MDL}(r_x,r_y,s) \nonumber \\
& = \underset{s=0,\ldots,r-1}{{\rm argmax}} [- I_{\rm MDL}(r_x,r_y,s)] \nonumber \\
& = \underset{s=0,\ldots,r-1}{{\rm argmax}} [I_{\rm MDL}(r_x,r_y,r) - I_{\rm MDL}(r_x,r_y,s)] \nonumber \\
& = \underset{s=0,\ldots,r-1}{{\rm argmax}} [\Lambda(r_x,r_y,s) - T_{\rm MDL}(r_x,r_y,s)].
\end{align}
The min-step in Detector 3 thus chooses the value of $s$ that {\em maximizes} the difference between the GLRT statistic and the MDL test threshold. This is different than the min-step in Detector 2, which picks the {\em smallest} $s$ for which the test statistic exceeds the threshold. Therefore, Detector 3 will never pick a $\hat{d}$ smaller, but possibly larger, than Detector 2.

\section{Performance evaluation}
\label{sec:sim}

In this section, we compare the performance of our three model-order selection schemes among each other and with competing approaches. In the absence of a competing systematic approach to the {\em joint} model-order selection in PCA-CCA, we determined the PCA ranks $r_x$ and $r_y$ separately from the number of correlated signals $d$. We used the sample eigenvalue-based (SEV) technique \cite{NadakuditiSPT0708} for selecting $r_x$ and $r_y$ because it is one of the few techniques that can handle the sample-poor case for a single channel. For the selection of $d$ we used the canonical correlation test (CCT) \cite{ChenIET-RSN1096} with $P_{\rm FA} = 0.005$, the Akaike information criterion (AIC) \cite{StoicaSPT0496}, and the MDL criterion \cite{StoicaSPT0496}. 

Figures \ref{fig_sim1}--\ref{fig_sim4} show the probability of selecting the correct $d$ for different setups. In the first setup, shown in Figs. \ref{fig_sim1}--\ref{fig_sim3}, we consider a system with $d = 2$ correlated signals (each with variance $5$ and correlation coefficients $0.8$ and $0.7$), and $f_x = 3$ and $f_y = 4$ independent signals (each with variance $1.5$). The matrices $\AB_x$ and $\AB_y$ are randomly generated unitary matrices. For each data point, we ran 1000 independent Monte Carlo trials.

\begin{figure}
\begin{center}
\includegraphics{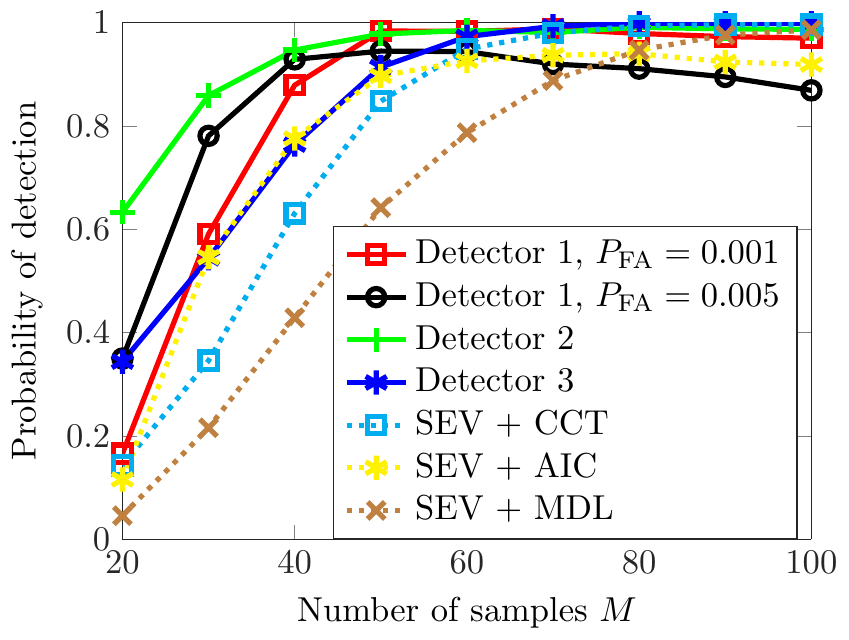}
\caption{Performance of our Detectors 1, 2, 3 and competing approaches for {\em white} noise. System dimensions are $m = n = 40$.}
\label{fig_sim1}
\end{center}
\end{figure}

\begin{figure}
\begin{center}
\includegraphics{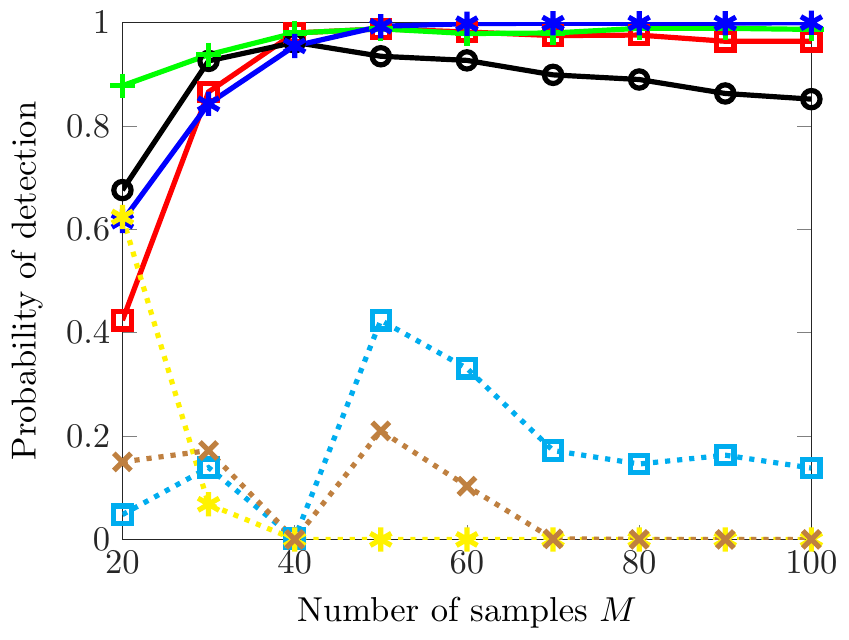}
\caption{Same setup as in Fig. \ref{fig_sim1} but with {\em colored} MA noise. For the meaning of the colored markers, please refer to the legend of Fig. \ref{fig_sim1}.}
\label{fig_sim2}
\end{center}
\end{figure}

We first consider a system with fixed dimension $m = n = 40$. In Fig. \ref{fig_sim1}, we show the performance as a function of the number of samples $M$ when the noise is white and each noise component has unit variance. We see that the performance of Detector 1 depends on the choice of $P_{\rm FA}$: For smaller $M$, $P_{\rm FA}$ should be chosen larger, whereas for larger $M$, a smaller $P_{\rm FA}$ performs better. Detector 2 does this trade-off automatically and performs very well even for very small sample sizes. All other approaches (including Detector 3) still perform quite well but require larger sample support. 

The picture completely changes when we have colored rather than white noise. We now generate the noise from a spatially varying moving average (MA) process of order $3$ with coefficients $[\frac{1}{\sqrt{3}}, \frac{1}{\sqrt{3}}, \frac{1}{\sqrt{3}}]$. Before the spatial averaging, the noise components have variance $1/3$. It can be seen in Fig. \ref{fig_sim2} that methods that select $r_x$ and $r_y$ separately from $d$ completely fail. This is because a single-channel technique such as SEV cannot distinguish between signal and noise eigenvalues if the noise is colored. The performance of our detectors, on the other hand, is actually improved particularly for very small sample sizes. 

\begin{figure}
\begin{center}
\includegraphics{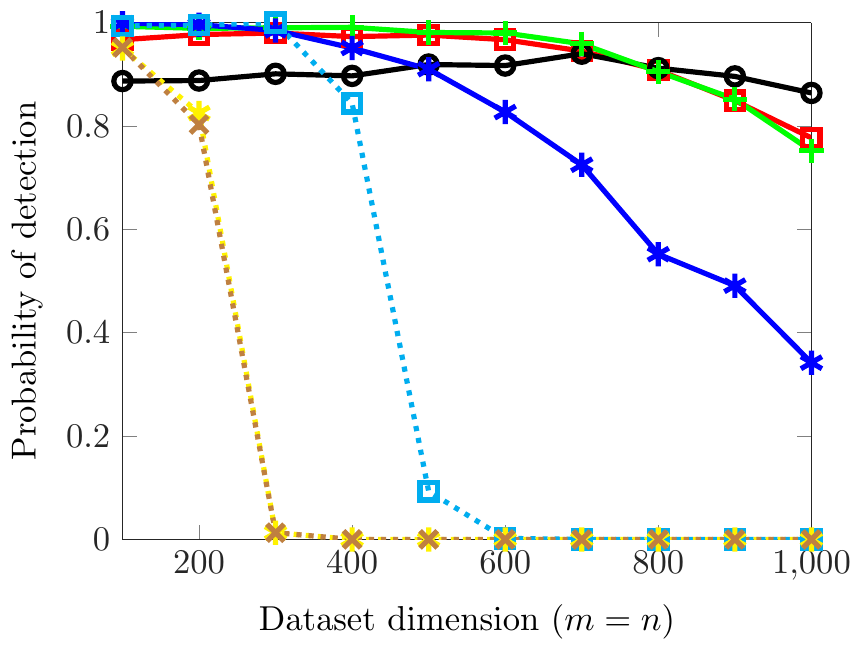}
\caption{Same setup as in Fig. \ref{fig_sim1} but with varying dimensions $m = n$ and fixed sample size $M = 100$. For the meaning of the colored markers, please refer to the legend of Fig. \ref{fig_sim1}.}
\label{fig_sim3}
\end{center}
\end{figure}

In Fig. \ref{fig_sim3}, we reconsider the white noise case but with varying dimensions $m = n$ and fixed sample size $M = 100$. When the noise is independent in space and time, increasing the ratio of the data dimensions $m$, $n$ to the number of samples $M$ shrinks the signal-subspace \cite{NadakuditiSPT0708}, which worsens the detection performance. We note, however, that the decrease in performance affects the SEV + X techniques much more than our detectors. Indeed, Detector 2 again shows a very reliable performance even for large dimensions. The main reason behind this effect is that the SEV technique is designed to keep all the signal components (i.e., correlated {\em and} independent components) whereas our detectors aim to eliminate weaker independent components in the PCA step. The presence of independent components deteriorates the detection performance of the subsequent CCA step.

{%\color{red}

\begin{figure}
\begin{center}
\includegraphics{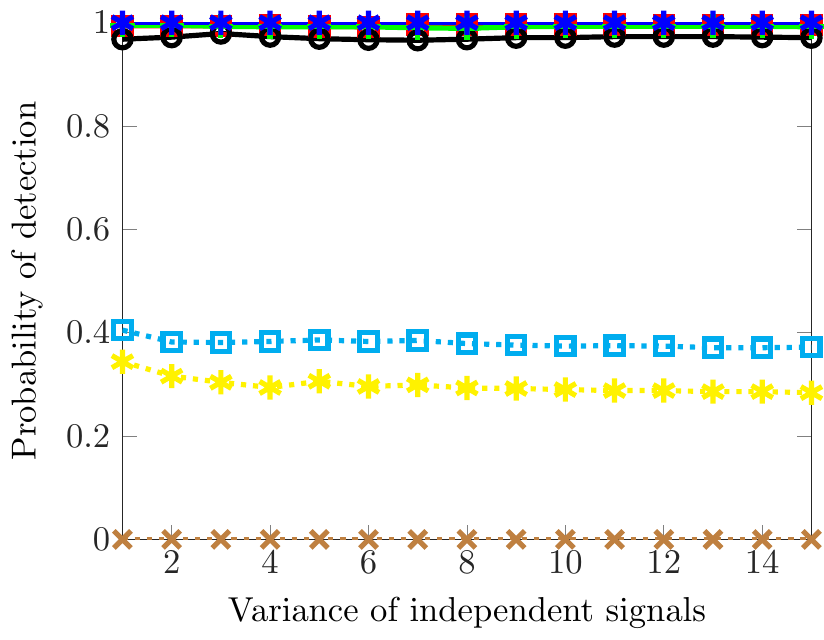}
\caption{Effect of the independent signals' variance on performance. Settings: $d = 7$ correlated signals with variance 10 and correlation coefficients ($0.92$, $0.9$, $0.88$, $0.85$, $0.83$, $0.8$, $0.75$), $f_x = f_y = 2$ independent signals of varying variance, $m =n = 80$, $M = 150$, colored AR(1) noise with coefficient $0.65$. For the meaning of the colored markers, please refer to the legend of Fig. \ref{fig_sim1}.}
\label{fig_new_scen_1}
\end{center}
\end{figure}

\begin{figure}
\begin{center}
\includegraphics{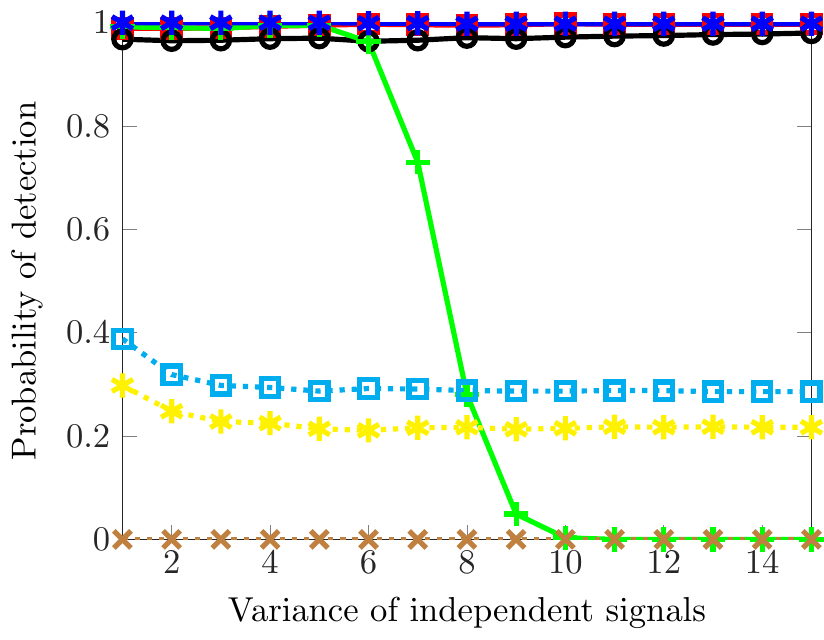}
\caption{Same setting as in Fig. \ref{fig_new_scen_1}, except that now there are $f_x = f_y = 4$ independent signals of varying variance. For the meaning of the colored markers, please refer to the legend of Fig. \ref{fig_sim1}.}
\label{fig_new_scen_2}
\end{center}
\end{figure}

So far, we have looked at a case where the correlated signals are {\em stronger} than the independent signals. We now investigate what happens when the correlated signals are {\em weaker} than some or all of the independent signals. First consider a scenario with 2 independent signals of varying variance and 7 correlated signals of variance 10. Figure \ref{fig_new_scen_1} shows the probability of detection as a function of the independent signals' variance. We see that the variance has only little effect on the performance of all techniques. 

Now we increase the number of independent signals to 4, leaving all other settings unchanged. The most dramatic effect that can be observed in Fig. \ref{fig_new_scen_2} is the failure of Detector 2 once the independent signals reach a variance close to the correlated signals' variance. This may be explained as follows. Detector 2 sets its threshold based on MDL, which generally does not overestimate the number of correlated signals, but may {\em underestimate} it if the sample size is not sufficiently large compared to the system dimension (i.e., the PCA rank). In the case shown in Fig. \ref{fig_new_scen_2}, there are 7 correlated signals and 4 independent signals. Once the independent signals become as strong as or stronger than the correlated signals, this leads to an optimum PCA rank of 11. As the number of samples $M = 150$ is not significantly larger than 11, MDL starts to underestimate the model order. This affects Detector 2 more severely than Detector 3 because Detector 3 will always return a model as large as, but possibly larger than, Detector 2 (see the discussion in Section \ref{sec_minMDL}). 

\begin{figure}
\begin{center}
\includegraphics{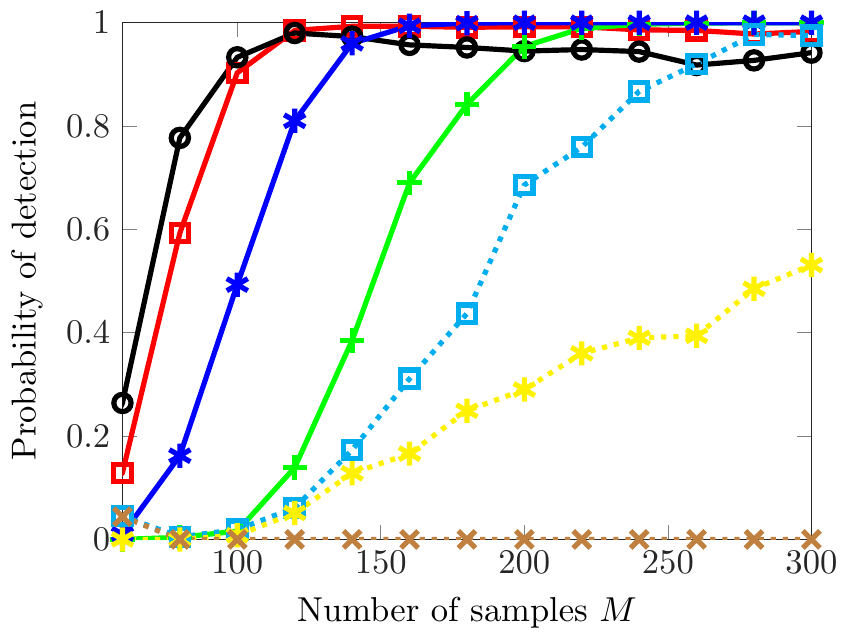}
\caption{Same setting as in Fig. \ref{fig_new_scen_1}, except that there are $d = 5$ correlated signals with variance 8, and $f_x = f_y = 7$ independent signals, two of which have variance 12 and 5 of which have variance 3. Performance as a function of number of samples $M$.  For the meaning of the colored markers, please refer to the legend of Fig. \ref{fig_sim1}.}
\label{fig_new_scen_3}
\end{center}
\end{figure}

This explanation can be validated by investigating the effect of the number of samples in a scenario where there are strong independent signals. We now consider a case with $d =5$ correlated signals with variance 8, and $f_x = f_y = 7$ independent signals, two of which have variance 12 and 5 of which have variance 3. In Fig. \ref{fig_new_scen_3}, we look at the performance as a function of the number of samples $M$. It can be observed that, among our three detectors, Detector 2 needs the largest number of samples for satisfactory performance, followed by Detector 3. The lesson that can be learned here is that in the presence of strong independent signals, Detector 1 should be preferred if only a very small number of samples are available.

\begin{figure}
\begin{center}
\includegraphics{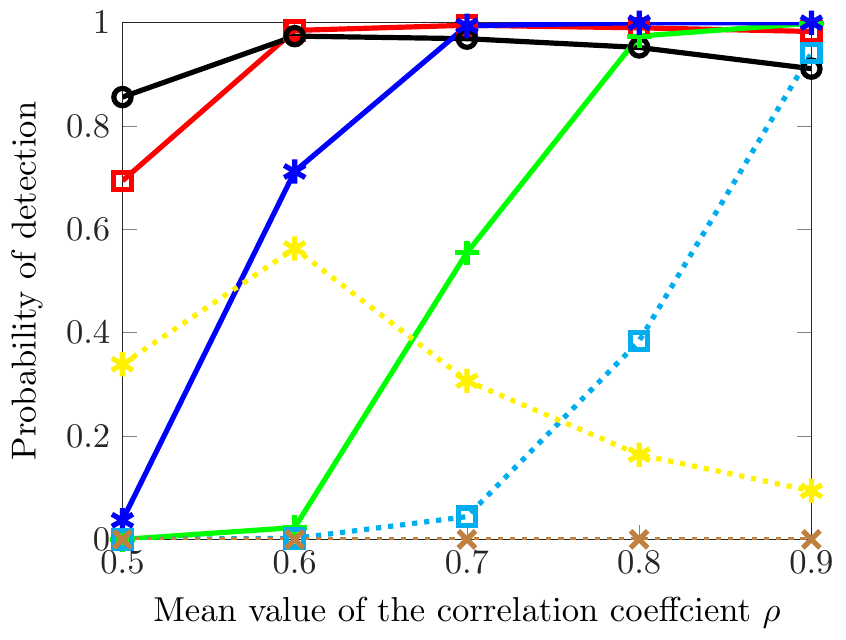}
\caption{Performance as a function of the mean correlation coefficient $\rho$. Settings: $m = n = 100$, $M = 180$, $d = 5$ correlated signals with correlation coefficients drawn from a uniform distribution between $[\rho - 0.05, \rho + 0.05]$, $f_x = f_y = 2$ stronger independent signals, AR(1) noise with coefficient $0.65$. For the meaning of the colored markers, please refer to the legend of Fig. \ref{fig_sim1}.}
\label{fig_new_scen_4}
\end{center}
\end{figure}

Let us now investigate the effect that the value of the correlation coefficients among the correlated signals have. Here we consider a scenario with $d = 5$ correlated signals with variance 8, and $f_x = f_y = 2$ stronger independent signals of variance 10. In Fig. \ref{fig_new_scen_4}, we plot the performance as function of $\rho$. The correlation coefficients for the 5 correlated signals are drawn from a uniform distribution between $[\rho - 0.05, \rho + 0.05]$. As expected, stronger correlation leads to better performance. Since the independent signals are stronger than the correlated signals, Detector 1 outperforms Detector 3, which in turn outperforms Detector 2. All of our detectors outperform the competition. }

\begin{figure}
\begin{center}
\includegraphics{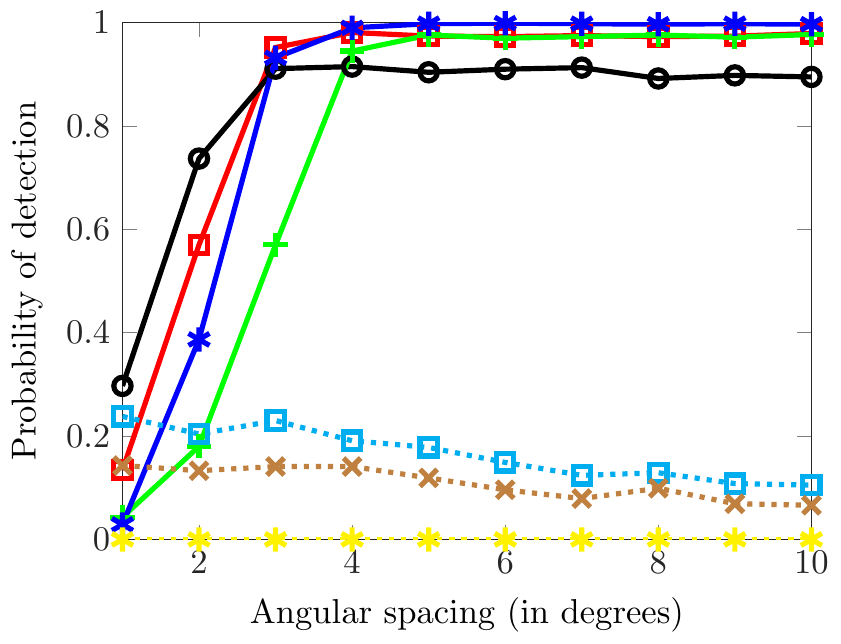}
\caption{Array processing toy example to illustrate the effect of ill-conditioned mixing matrices. For the meaning of the colored markers, please refer to the legend of Fig. \ref{fig_sim1}.}
\label{fig_sim4}
\end{center}
\end{figure}

In our last example we examine an array processing toy application to see what happens if the mixing matrices $\AB_x$ and $\AB_y$ become ill-conditioned. We consider two spatially separated uniform linear arrays (ULAs) with 40 sensors (i.e., $m = n = 40$) and inter-sensor spacing of $\lambda/2$, which take $M = 60$ samples. There are $5$ fixed point-sources in the far-field emitting narrow-band Gaussian signals at wavelength $\lambda$, which impinge upon ULA~1 at angles $[\theta_{x,1},\theta_{x,2},\ldots,\theta_{x,5}] = [20^\circ, 20^\circ + \delta, ..., 20^\circ + 4\delta]$. Similarly, $6$ such signals impinge upon ULA~2 at angles $[\theta_{y,1},\theta_{y,2},\ldots,\theta_{y,6}] = [50^\circ, 50^\circ + \delta, ..., 50^\circ + 5\delta]$. Two of these signals are correlated between ULAs 1 and 2 (i.e., $d = 2$, $f_x = 3$, $f_y = 4$) with correlation coefficients $0.8$ and $0.7$. The correlated signals each have variance $5$ and the independent signals each have variance $1.5$. The noise is colored and generated as in the setup for Fig.~\ref{fig_sim2}. 

With these assumptions, the $i$th column of $\AB_x$ is $[1, e^{j\frac{\pi}{2} \sin\theta_{x,i}}, ... , e^{j\frac{\pi}{2} (n-1) \sin\theta_{x,i}}]^T$, $i = 1, \ldots, 5$, and the $i$th column of $\AB_y$ is $[1, e^{j\frac{\pi}{2} \sin\theta_{y,i}},\ldots, e^{j\frac{\pi}{2} (m-1)\sin\theta_{y,i}}]^T, i=1,\ldots,6$. As the angular spacing $\delta$ decreases, the mixing matrices ${\bf A}_x$ and ${\bf A}_y$ become more ill-conditioned. Figure \ref{fig_sim4} shows the performance of all detectors for angular spacing $\delta$ ranging from $1^\circ$ to $10^\circ$. We can see that due to the presence of colored noise, all SEV + X methods fail irrespective of $\delta$. Our detectors, on the other hand, are able to provide very good detection rates from $\delta=4^\circ$ onward. Detectors 1 and 3 provide the best results for small $\delta$.

\section{Conclusions}

PCA-CCA is a common approach to the analysis of correlation between two data sets when there is only small sample support. In the past, selecting the ranks of the PCA steps and identifying the number of correlated signals was often done by ad-hoc rules or based on experience. In this paper, we have presented a systematic approach to the joint order selection of PCA ranks and number of correlated signals, based on a GLRT and information-theoretic criteria. Simulation results have shown that the techniques perform very well for extremely sample-poor scenarios in particular in the presence of colored noise. Of course, it is important to remember that there is no free lunch. While we do not need many samples compared to the dimensions of the data sets, the techniques do require the number of samples to be sufficiently greater than the sum of the numbers of correlated signals and stronger independent signals (i.e., variance larger than the correlated signals).

\appendices
\section{Effect of PCA on estimated canonical correlations}
\label{app:interlace}

\begin{lemma}\label{lemma_interlace}
The estimated canonical correlation coefficients increase with increasing PCA ranks $r_x$ and $r_y$: $\hat{k}_i(\tilde{r}_1,\tilde{r}_2) \geq \hat{k}_i(r_x,r_y) ,i=1,\ldots,\min(r_x,r_y)$, for $1\leq r_x < \tilde{r}_1$ and $1\leq r_y < \tilde{r}_2$.
\end{lemma}

\begin{proof}
Define the following matrices:
 \begin{eqnarray}
 {\bf G} &=& {\bf V}^H_x(:,1:\tilde{r}_1) {\bf V}_y(:,1:\tilde{r}_2) = \left[\begin{array}{c} {\bf G}_1 \\ {\bf G}_2\end{array}\right], \nonumber\\
 {\bf G}_{1} &=& {\bf V}^H_x(:,1:r_x) {\bf V}_y(:,1:\tilde{r}_2) = \left[\begin{array}{cc} {\bf G}_{1,1} & {\bf G}_{1,2}\end{array}\right], \nonumber\\
 {\bf G}_{2} &=& {\bf V}^H_x(:,r_x+1:\tilde{r}_1) {\bf V}_y(:,1:\tilde{r}_2), \nonumber\\
 {\bf G}_{1,1} &=& {\bf V}^H_x(:,1:r_x) {\bf V}_y(:,1:r_y), \nonumber\\
 {\bf G}_{1,2} &=& {\bf V}^H_x(:,1:r_x) {\bf V}_y(:,r_y+1:\tilde{r}_2). \nonumber
 \end{eqnarray}
 
According to the Cauchy interlacing theorem, we have
\[\lambda_i\left({\bf GG}^H\right)=
\lambda_i\left(\left[
\begin{array}{cc}
{\bf G}_{1}{\bf G}_{1}^H & {\bf G}_{1}{\bf G}_{2}^H \\
{\bf G}_{2}{\bf G}_{1}^H & {\bf G}_{2}{\bf G}_{2}^H
\end{array}\right]\right)\geq
\lambda_i\left({\bf G}_{1}{\bf G}_{1}^H\right)\] for $i=1,\ldots,r_x$, where $\lambda_i(\cdot)$ represents the $i$th largest eigenvalue.
Furthermore, as a result of the Weyl inequality, we also have
$\lambda_i\left({\bf G}_{1}{\bf G}_{1}^H\right)=\lambda_i\left({\bf G}_{1,1}{\bf G}_{1,1}^H+{\bf G}_{1,2}{\bf G}_{1,2}^H\right)\geq \lambda_i\left({\bf G}_{1,1}{\bf G}_{1,1}^H\right)$.
Together with first inequality, this yields
$\lambda_i\left({\bf GG}^H\right) \geq \lambda_i\left({\bf G}_{1,1}{\bf G}_{1,1}^H\right)$.
As ${\bf V}_x$ and ${\bf V}_y$ represent the matrices of right-singular vectors of $\XB$ and $\YB$, respectively, it follows that the squared sample canonical correlation coefficient $\hat{k}_i^2(\tilde{r}_1,\tilde{r}_2) = \lambda_i\left({\bf GG}^H\right)$ is greater than or equal to the squared sample canonical correlation coefficient  $\hat{k}_i^2(r_x,r_y) = \lambda_i\left({\bf G}_{1,1}{\bf G}_{1,1}^H\right)$.
\end{proof}

\bibliographystyle{IEEEtran}
\bibliography{refs}

\end{document}